\documentclass[12pt,a4paper]{article}

\usepackage{natbib}
\usepackage{amssymb}
\usepackage{hyperref}       
\usepackage{url}            
\usepackage{booktabs}       
\usepackage{amsfonts}       
\usepackage{nicefrac}       
\usepackage{microtype}      
\usepackage{lipsum}		
\usepackage{graphicx}
\usepackage{braket}

\usepackage{fancyvrb}
\usepackage{amsthm}
\usepackage{amsmath}
\usepackage{algorithmic}
\usepackage[titlenumbered,linesnumbered,ruled,vlined]{algorithm2e}

\theoremstyle{plain}

\newcommand\pauli[1]{\sigma^z_{#1}}

\usepackage{tikz}
\usetikzlibrary{arrows}
\usetikzlibrary{shapes}
\usetikzlibrary{positioning} 
\usetikzlibrary{fit}
\usetikzlibrary{shadows} 
\usetikzlibrary{calc} 
\usetikzlibrary{decorations.markings}
\usetikzlibrary{arrows.meta, bending, intersections}

\tikzstyle{arrowblue}=[-,line width=1pt,blue]
\tikzstyle{arrowred}=[-,line width=1pt,red, dashed]
\tikzstyle{nodo}=[draw=black, line width=1pt, ellipse]

\tikzset{splrect/.style = {rectangle split, rectangle split horizontal,
                           rectangle split parts=5, minimum height=1cm, 
                           align=center, draw=black},
         rect/.style = {rectangle, rounded corners, minimum width=3cm,
                        minimum height=1cm,text centered, text width=3cm,
                        draw=black},
         rect2/.style = {rectangle, rounded corners, minimum width=4cm,
                        minimum height=1cm,text centered, text width=4cm,
                        draw=black},
         redrect/.style = {rectangle, rounded corners, minimum width=3cm,
                           minimum height=1cm,text centered, text width=3cm,
                           draw=black, fill=red, fill opacity=0.5},
         freerect/.style = {rectangle, rounded corners, minimum width=10cm, minimum height=1cm, draw=black, text centered, text width=10cm},
         arrow/.style = {->,shorten >=1pt,>=stealth',semithick},
         labl/.style = {minimum width=3cm,
                        minimum height=1cm,text centered, text width=3cm,
                        draw=black!0},
         doc/.style={draw, minimum height=4em, minimum width=3em, text width=3cm, text centered, fill=white, double copy shadow={shadow xshift=4pt, shadow yshift=4pt, fill=white, draw}},
         decision/.style={diamond, minimum width=3cm, minimum height=1cm, text centered, draw=black}
}

\tikzset{
  invisible/.style={opacity=0},
  visible on/.style={alt={#1{}{invisible}}},
  alt/.code args={<#1>#2#3}{%
    \alt<#1>{\pgfkeysalso{#2}}{\pgfkeysalso{#3}} 
  },
}

\newcommand{\eclause}{{\unitlength 2.3mm \,\framebox(1,1){}\,}}
\newcommand{\wcla}[2]{_{(#1)}\,#2}
\newcommand{\xor}{\oplus}
\newcommand{\ignorar}[1]{}
\newcommand{\opt}{Opt}
\newcommand{\cost}{Cost}
\newcommand{\weight}{Weight}

\usepackage{url}
\usepackage{xcolor}

\newtheorem{example}{Example}
\newtheorem{theorem}{Theorem}
\newtheorem{lemma}[theorem]{Lemma}
\newtheorem{corollary}[theorem]{Corollary}

\newtheorem{definition}[theorem]{Definition}

\theoremstyle{definition}

\theoremstyle{remark}

\newtheorem*{note*}{Note}
\newtheorem{remark}[theorem]{Remark}
\newtheorem*{remark*}{Remark}

\newtheorem*{claim*}{Claim}

\title{SAT, Gadgets, Max2XOR\\and Quantum Annealers}

\author{Carlos Ansotegui\\
Logic \& Optimization Group,\\ Universitat de Lleida, Spain
\and
Jordi Levy\\
IIIA, CSIC, Spain}

\begin{document}
\maketitle

\begin{abstract}
Quantum Annealers are presented as quantum computers that with high probability can optimize certain quadratic functions on Boolean variables in constant time. These functions are basically the Hamiltonian of Ising models that reach the ground energy state, with a high probability, after an annealing process. They have been proposed as a way to solve SAT in some preliminary works.

These Hamiltonians can be seen as Max2XOR problems, i.e. as the problem of finding an assignment that maximizes the number of XOR clauses of at most 2 variables that are satisfied. In this paper, we focus on introducing several gadgets to reduce SAT to Max2XOR. We show how they can be used to translate SAT instances to initial configurations of a quantum annealer.
\end{abstract}







\section{Introduction}
\subsection*{Quantum Annealers}

Quantum Computation exploits some particularities of quantum mechanics, such as superposition, interference, and entanglement to solve some hard problems such as the ones we face in Artificial Intelligence. Current quantum computers are still too small to compete with classical computers, but the rapid increase in their number of qubits could replace them in some of these hard problems. Most research in quantum computing is focused on the model of quantum circuits, where gates are quantum operators that act on a qubit or a superposition of them. 

An alternative and equivalent~\cite{equivalentquantum} model is \emph{Quantum Annealing}. D-Wave Systems Inc.\footnote{\url{https://www.dwavesys.com/}.} has constructed some computers based on this model of computation with more than 5000 qubits, surpassing (in the number of qubits) the rest of the manufacturers of quantum computers. In these machines, we can only decide the \emph{coupling factor} $J_{ij}$ between certain pairs of qubits $i$ and $j$, and the \emph{biases} $h_i$ of each qubit $i$. 
From an initial situation described by the Hamiltonian $\mathcal{H} = \sum_i \sigma^x_i$, the system evolves into a system described by an Ising Hamiltonian 
\[
\mathcal{H}_{Ising} = \sum_{i\in V_G} h_i \pauli{i} + \sum_{(i,j)\in E_G} J_{ij}\,\pauli{i}\,\pauli{j}
\]
where $\pauli{i}$ are the z-Pauli operator acting on qubit $i$ and $G=(V_G,E_G)$ is a graph describing the architecture of the machine.
During the process, the system is tried to be kept in the lowest-energy eigenstate, thus at the end, ideally, we get the lowest-energy eigenstate of the Ising model that encodes a minimization problem. This final state is classical, without superpositions.
The probability of success depends on the difference of energy between the lowest-energy state and the next energy level. The Landau-Zener model describes the probability of success (finishing in the minimal energy state) for a two-states system with one $1/2$-spin particle with energy gap $\Delta$ under the action of a varying external magnetic field at rate $c$ as
\(
p = 1-e^{\frac{-\pi\,\Delta^2}{4c}}
\).
\citet{hard2SAT} argue that, although the quantum annealers are more complicated systems, the probability can be approximated by this two-level system.

In order to increase the probability of finding this minimum, the experiment is repeated many times.
The ranges of biases and couplings depend on the particular model of the quantum annealer.\footnote{For old Chimera-base D-Wave 2000 series, they were $h_i\in[-2,2]$ and $J_{ij}\in[-1,1]$. In more modern models they are $h_i\in[-4,4]$ and $J_{ij}\in[-2,1]$.} In this paper, for simplicity, we assume that all of them are in the range $h_i\in[-1,1]$ and $J_{ij}\in[-1,1]$, being our conclusions easily adapted for other ranges. The structure of the machine also depends on the model. For our purposes, we simply assume that $G$ is an undirected sparse graph. In other words, we will idealize quantum annealers as devices that are able to solve --in constant time and with high probability~\footnote{Here we use high probability in the mathematical sense, as a non-zero probability.}  that depends exponentially on the difference between the minimal solution and the next state-- a quadratic function where only some quadratic terms are allowed.

In this paper, we will show that (obviating some technical details) a quantum annealer can be seen as a Max2XOR solver with constraint weights in the reals $[0,1]$. Therefore, we focus on the problem of finding gadgets to reduce the satisfiability of a CNF formula (SAT) to the satisfiability of the \emph{maximum} number of 2XOR constraints (Max2XOR) and analyze which (theoretical) features of the gadgets are desirable for quantum annealers. In some sense, a subproduct of our contribution could be seen as one more step to assess the maturity and potential of quantum annealers to solve the SAT problem.

\subsection*{Max2XOR, QUBO and Ising model}

Exclusive OR (XOR), here written $\xor$, may be an alternative to the use of traditional OR to represent propositional formulas. In practice, many approaches combining SAT and XOR reasoning have been presented \cite{Li00,Li00prima,BaumgartnerM00,Li03,HeuleM04,HeuleDZM04,Chen09,SoosNC09,LaitinenJN12sat12,Soos10,LaitinenJN11,LaitinenJN12ictai12,SoosM19}. By writing clauses $x_1\xor \cdots \xor x_k$ as constraints $x_1\xor \cdots \xor x_k=1$ or $x_1\xor \cdots \xor x_k=0$, where $1$ means true and $0$ false, we can avoid the use of negation, because $\neg x\xor C=k$ is equivalent to $x\xor C=1-k$. The equivalent to the resolution rule for XOR constraints, called XOR resolution rule, is
\[
\begin{array}{c}
x\xor A=k_1\\
x\xor B=k_2\\
\hline
A\xor B = k_1\xor k_2
\end{array}
\]
where $A$ and $B$ are clauses. In the particular case of $A=B=\emptyset$ and $k_1\neq k_2$, this rule concludes the contradiction $0=1$, that we represent as $\eclause$. The proof system containing only this rule allows us to produce polynomial refutations for any unsatisfiable set of XOR constraints, using Gaussian elimination. Therefore, unless $P=NP$, we cannot polynomially translate any propositional formula into an equivalent conjunction of XOR constraints.\footnote{We can translate any OR clause $x_1\vee \dots\vee x_k$ into the set of weighted XOR constraints:
\(
\bigcup_{S\subseteq \{1,\dots,k\}} \{\wcla{1/2^{k-1}}{\bigoplus_{i\in S} x_i = 1}\}
\)
This translation allows us to reduce SAT to MaxXOR. However, the reduction is not polynomial, because it translates every clause of size $k$ into $2^k-1$ constraints. In Section~\ref{sec:MaxSAT->Max2XOR}, we describe a polynomial reduction that avoids this exponential explosion on expenses of introducing new variables.} 
However, the SAT problem can be reduced to the Max2XOR problem, i.e., the problem of maximizing the number of 2XOR constraints that can be satisfied simultaneously. 

Formally, a Max2XOR problem is a set of pairs \[\{\wcla{w_1}{C_1=k_1},\dots,\wcla{w_n}{C_n}= k_n\},\] where $w_i$ are positive rational numbers (we can also consider the special case of natural weights or the generalization to positive real weights) representing the penalty for violating the constraint, $C_i$ are clauses form by one variable or the XOR disjuntion (or sum modulo two) of two variables, and $k_i$ are either $0$ or $1$, representing $false$ and $true$, respectively. A solution is an optimal assignment that maximizes the sum of the weights of satisfied constraints. 
For example, $\{\wcla{1/2}{x=1},\ \wcla{1/2}{y=1},\ \wcla{1/2}{x+y=1}\}$ is a Max2XOR problem where the set of solutions (optimal assignments) are exactly the set of solutions of the SAT formula $\{x\vee y\}$, i.e. the assignments that assign $x$, or $y$, or both, to one. All these three optimal assignments satisfy two constraints with a total weight equal to one.

A Quadratic Unconstrained Binary Optimization (QUBO) problem is a minimization problem:
\[
\min_{x_i\in\{0,1\}} \sum_i a_i x_i + \sum_{i<j}b_{ij}x_i x_j
\]
Notice that any Max2XOR problem $\Gamma$ may be translated into an equivalent QUBO problem, taking:
\[
\begin{array}{l}
\displaystyle a_i = \sum_{\wcla{w}{x_i=0}\in \Gamma}\kern-4mm w \ - \sum_{\wcla{w}{x_i=1}\in \Gamma} \kern-4mm w\ +
\sum_{\wcla{w}{x_i+y=0}\in \Gamma} \kern-4mm w\ - \sum_{\wcla{w}{x_i+y=1}\in \Gamma} \kern-4mm w\\[8mm]
\displaystyle b_{ij} = -2\kern-6mm\sum_{\wcla{w}{x_i+x_j=0}\in \Gamma} \kern-4mm w \ +\  2 \kern-6mm\sum_{\wcla{w}{x_i+x_j=1}\in \Gamma} \kern-4mm w
\end{array}
\]
Conversely, any QUBO problem may be easily translated into a Max2XOR problem (where all weights $w$ are positive, whereas $a_i$ and $b_{ij}$ may be negative).

Any QUBO problem can be translated into an Ising formulation
\[
\min_{z_i\in\{-1,+1\}}\sum_i h_i z_i + \sum_{i<j} J_{ij} z_i z_j
\]
via the bijection $h_i = \frac{1}{2}a_i +\frac{1}{4}\sum_j b_{ij}$, $J_{ij}=\frac{1}{4}b_{ij}$, and similarly for any Max2XOR problem $\Gamma$ via $h_i = \frac{1}{2}(\sum_{\wcla{w}{x_i=0}\in \Gamma} w \ - \sum_{\wcla{w}{x_i=1}\in \Gamma} w)$ and $J_{ij}=\frac{1}{2}(\sum_{\wcla{w}{x_i+x_j=1}\in \Gamma} w \ -\  \sum_{\wcla{w}{x_i+x_j=0}\in \Gamma}w)$.
Therefore, Max2XOR, QUBO, and Ising model annealing are three NP-complete equivalent problems. 
For example,
\begin{itemize}
    \item 
the Max2XOR problem $\{\wcla{1}{x_1=0},\ \wcla{1}{x_1+x_2=0}\}$, 
\item the QUBO problem $\min_{x_1,x_2\in\{0,1\}}2x_1+x_2-2x_1x_2$ and 
\item the Ising problem $\min_{z_1,z_2\in\{-1,+1\}}\frac{1}{2}z_1 -\frac{1}{2}z_1z_2$ 
\end{itemize}
are all of them equivalent. In all these problems we can multiply all weights/coefficients by the same constant, obtaining an equivalent problem. However, we define the transformations to preserve the \emph{energy gap} between the lowest-energy state and the next classical state.

We can see quantum annealers as hardware to try to some this problem in any of these three presentations.
In this paper, to formalize the reduction from SAT to Max2XOR (or equivalently to QUBO or Ising) we will use gadgets.

\ignorar{
We will consider weighted constraints, noted $\wcla{w}{C=k}$, where $w$ is a rational number that denotes the contribution of the satisfiability of the constraint to the formula. This way, we can translate every binary clause $x\vee y$ as $\{\wcla{1/2}{x=1},\ \wcla{1/2}{y=1},\ \wcla{1/2}{x\xor y=1}\}$, because when $x$ or $y$ are equal to one (i.e. $x\vee  y$ is satisfied), exactly two of the XOR constraints are satisfied, which contributes $1/2+1/2=1$ to the constraint, and when $x$ and $y$ are both equal to zero and the original clause is falsified, none of the XOR constraints are satisfied. We can also translate ternary clauses like $x\vee y\vee z$ as $\{\wcla{1/4}{x=1},$ $\wcla{1/4}{y=1},$ $\wcla{1/4}{z=1},$ $\wcla{1/4}{x\xor y=1},$ $\wcla{1/4}{x\xor z=1},$ $\wcla{1/4}{y\xor z=1},$ $\wcla{1/4}{x\xor y\xor z=1}\}$. 

In Section~\ref{sec:OnProofSystems}, we will discuss the possible definition of a proof system for Max2XOR, in the spirit of the MaxSAT resolution which was first defined by~\cite{LarrosaH05}, and proven complete by~\cite{SAT06,AIJ1}.
}

\subsection*{Gadgets}
Traditionally, the word "gadget" is used to denote a finite combinatorial structure that allows translating constraints of one
problem to constraints of another.
In~\cite{gadgets}, the notion is formalized, defining a $(\alpha,\beta)$-\emph{gadget} as a function from a family of constraints $\mathcal{F}_1$ to another family $\mathcal{F}_2$ that translate every constraint $f(\vec{x})$ in $\mathcal{F}_1$ to a set of $\beta$ constraints $\{g_i(\vec{x},\vec{b})\}_{i=1,\dots,\beta}$ in $\mathcal{F}_2$, where the $b$'s are (fresh) \emph{auxiliary variables} such that, when $f(\vec{x})$ is satisfied for some assignment of the $x$'s, we can find an assignment to the $b$'s that satisfy $\alpha$ constraints $g(\vec{x},\vec{b})$. Conversely, when $f(\vec{x})$ is falsified, no assignment for the $b$'s satisfy strictly more than $\alpha-1$ constraints $g(\vec{x},\vec{b})$.
If, additionally, when $f(\vec{x})$ is falsified, some assignment for the $b$'s satisfy exactly $\alpha-1$ constraints $g(\vec{x},\vec{b})$, we say that the gadget is \emph{strict}.
This definition can be generalized to \emph{weighted} constraints, where the weights are real numbers.\footnote{Although a more detailed analysis would show that we can restrict them to be rational numbers.} Then, instead of the number of satisfied constraints, we talk about the sum of the weights of satisfied constraints.

There is a long tradition of the use of gadgets. For instance, the $(k-2,k-2)$-gadget:
\begin{equation}
x_1\vee\cdots\vee x_k \to \left\{ \begin{array}{l}
x_1\vee x_2\vee b_1,\\
\neg b_1\vee x_3\vee b_2,\\
\cdots,\\
\neg b_{k-4}\vee x_{k-2}\vee b_{k-3},\\
\neg b_{k-3}\vee x_{k-1}\vee x_k
\end{array}\right.
\label{eq-MaxSAT->Max3SAT}
\end{equation}
reduces $k$SAT to 3SAT.
The following (weighted) $(3.5,4)$-gadget~\cite{gadgets} reduces 3SAT to Max2SAT:
\begin{equation}
x_1\vee x_2\vee x_3 \to \left\{
\begin{array}{ll}
\wcla{1/2}{x_1\vee x_3}, 
\ \wcla{1/2}{\neg x_1\vee \neg x_3},\\
\wcla{1/2}{x_1\vee \neg b},
\ \wcla{1/2}{\neg x_1\vee b},\\
\wcla{1/2}{x_3\vee \neg b},
\ \wcla{1/2}{\neg x_3\vee b},\\
\wcla{1}{x_2\vee b}
\end{array}\right.
\label{eq-Max3SAT->Max2SAT}
\end{equation}

In \cite{AnsoteguiL21} new gadgets reducing SAT to Max2SAT are introduced. The authors show, formally and empirically, that the new reduction allows solving efficiently the Pigeon Hole Principle, i.e. they prove that there exists a polynomial MaxSAT resolution proof~\cite{AIJ1}  and that a general purpose MaxSAT solvers can solve the resulting Max2SAT formula.

Defining $\opt(P)$ as the maximal number of satisfied constraints of a problem $P$, $\cost(P)$ as the minimal number of falsified constraints, and $\weight(P)$ the number of constraints in $P$, we have that,
if $P$ is translated into $P'$ using a $(\alpha,\beta)$-gadget, then
\[
\begin{array}{l}
\opt(P')\leq(\alpha-1)\weight(P)+\opt(P)\\
\cost(P')\geq(\beta-\alpha)\weight(P)+\cost(P)
\end{array}
\]
with equalities if the gadget is strict.
This allows us to obtain an algorithm to maximize $P$, using an algorithm to maximize $P'$.

When $\alpha=\beta$ we preserve the cost, hence the satisfiability. Then, given a decision algorithm for $P'$, we get a decision algorithm for $P$.

Moreover, when we have a $p$-approximation algorithm for $P'$ we can get an $(1-\alpha(1-p))$-approximation algorithm for $P$. Therefore, traditionally, we were interested in gadgets with minimal~$\alpha$ in order to minimize the error in the approximation.

If the algorithm is based on deriving empty clauses, i.e. that prove lower bounds for the cost, we will be interested in gadgets that minimize $\beta-\alpha$. Notice that in the case that $P$ is a decision problem, then we can say that $P$ is unsatisfiable if, and only if, $\cost(P') \geq (\beta-\alpha)\weight(P)+1$.

In this paper, we are interested in maximizing another feature of gadgets, that we will call \emph{energy gap}, and will be introduced later. We will also discuss on the number of auxiliary variables the gadget introduces and how \emph{flexible} is the structure itself of the gadget.

\subsection*{Related Work}

Santra et~al.~\cite{quantumMax2SAT} experimentally study the performance of one of the first quantum annealers, with 108 qubits, on random Max2SAT problems. Despite the small size of tested instances, they observe that the probability of obtaining a correct answer decreases with the clause density (when we increase the number of clauses while keeping the number of variables fixed). This is explained by the decrease in the energy gap between the ground state and the first excited state. They find that success probability decreases around the critical clause density. However, it is not correlated with the time required by classical MaxSAT solvers (they compare \emph{akmaxsat}) to find a solution. 
More in detail, they identify every Boolean variable $x_j$ with a qubit $j$, the value \emph{true} with the eigenstate $1$ of the Pauli spin operator $\pauli{j}$ acting on qubit $j$, and \emph{false} with the eigenstate $-1$, i.e. $\pauli{j}\ \Ket{x_j\!=\!\mbox{true}} = \Ket{x_j\!=\!\mbox{true}}$ and $\pauli{j}\ \Ket{x_j\!=\!\mbox{false}} = -\Ket{x_j\!=\!\mbox{false}}$.
To avoid the problem with not connectivity between all pairs of qubits, they only consider Max2SAT formulas where clauses only contain pairs of variables whose corresponding qubits can be coupled. Using the Hamiltonian
\(
H = \sum_{(s\cdot x_i\vee s'\cdot x_j)\in\Gamma} \frac{\mathbb{I}- s\cdot\pauli{i}}{2}\frac{\mathbb{I}-s'\cdot\pauli{j}}{2}
\),
where $s,s'\in\{1,-1\}$ represent the sign of the variable in the clause, we get an energy penalty of $1$ for every clause violated by an assignment (the energy in an observed state is equal to the number of clauses violated by the assignment it represents). This Hamiltonian can be encoded with the biases $h_i = -1/4\sum_{(s\cdot x_i\vee s'\cdot x_j)\in\Gamma} s$ and couplings $J_{ij} = 1/4\sum_{(s\cdot x_i\vee s'\cdot x_j)\in\Gamma} s\cdot s'$. To ensure that all of them are in the range $[-1,1]$, we have to re-scale the Hamiltonian, multiplying by $1/ \max\{\max_i h_i,\max_{i,j} J_{ij}\}$. In the best case, when the density of clauses is low, this factor is $4$ which results in an energy penalty of $4$ for every violated clause. But, for high clause densities, the factor is smaller which results in a lower precision in the method.

Chancellor et al.~\cite{Chancellor} study the translation of kSAT and parity problems to quantum annealing. They observe that $x_1\vee x_2\vee x_3$ can be reduced to $x_1+x_2+x_3-x_1x_2-x_1x_3-x_2x_3+x_1x_2x_3$, or in general, $x_1\vee\cdots\vee x_k$ to $1+\sum_{S\subseteq \{1,\dots,k\}} (-1)^{|S|+1}\prod_{i\in S} x_i$. There is no problem with linear and quadratic terms, representing variables as qubits. For terms of bigger size, they use a new qubit to encode their value. However, this method has two disadvantages, one, generate couplings between all pairs of qubits, and second, the number of additional qubits grow exponentially with $k$.

N\"u\ss lein et al.~\cite{Nusslein} evaluate the performance of some methods to solve/translate 3SAT problems with/to quantum annealing. They also propose a new reduction and conclude that their method is better than Chancellor et al.~\cite{Chancellor}, and also than Choi~\cite{Choi}. However, they do not compare with Bian et al.~\cite{quantumSebastianiConference,quantumSebastianiJournal} and the new gadgets presented in this paper.

Douglass et al.~\cite{SATfilters} explores the use of a quantum annealer for the construction of SAT filters. In this case, the encoding involves the translation of MaxCUT, Not-all-equal 3-SAT,
2-in-4-SAT into a Hamiltonian. This translation is almost direct. The authors already remarked that the translation of SAT would involve the use of auxiliary variables, which are called \emph{ancillary variables}.

Bian et al.~\cite{quantumSebastianiConference,quantumSebastianiJournal} already using a 2048 qubits quantum annealer, explore the feasibility of solving SAT, without limitation on the size of clauses. The basic idea is to decompose any clause (or other subformulas) into constraints of at most 3 variables using the traditional Tseitin encoding. Like in the case of gadgets, this requires the introduction of auxiliary variables. For instance, we can decompose $x_1\vee x_2\vee x_3\vee x_4$ as $\{b_1\leftrightarrow x_1\vee x_2, b_2\leftrightarrow x_3\vee x_4, b_1\vee b_2\}$. Then, each one of these constraints or subformulas contributes to the Hamiltonian. For instance, $b_1\leftrightarrow x_1\vee x_2$ contributes by adding $5/2+1/2x_1+1/2x_2-b_1+1/2x_1x_2-x_1b_1-x_2b_1$\footnote{Abusing from the notation, we identify the variable (that takes Boolean values $0$ or $1$) with the qubit representing it (that takes values +1 or -1 when measured).} to the Hamiltonian. Below, we will see that this decomposition can in fact be interpreted as a gadget, although the authors do not mention the notion of a gadget in their paper. To ensure that the coupling factors are in the range $[-1,1]$, in some cases they have to re-normalize the entire Hamiltonian, dividing all factors by a constant. Moreover, they have to use several qubits to represent the same variable, when there are several occurrences of a variable in a formula. Then, they have to ensure that there is a path of coupled qubits that connect all these occurrences. They present some heuristics to find these paths.

Choi~\cite{minimalEmbedding} analyzes the problem of allocating variables into qubits and models it as a variant of graph embedding where in the process some transformations are allowed, basically the contractions of paths into edges. Bian et~al.~\cite{discreteOptimization} also analyze some methods with the same purpose.

\section{Preliminaries}

A k-ary \emph{constraint function} is a Boolean function $f:\{0,1\}^k\to\{0,1\}$. 
A \emph{constraint family} is a set $\mathcal F$ of constraint functions (with possibly distinct arities). 
A \emph{constraint}, over variables $V=\{x_1,\dots,x_n\}$ and constraint family $\mathcal F$, is a pair formed by a k-ary constraint function $f\in\mathcal{F}$ and a subset of $k$ variables, noted $f(x_{i_1},\dots,x_{i_k})$, or $f(\vec{x})$ for simplicity.  
A \emph{(weighted) constraint problem} or \emph{(weighted) formula} $P$, over variables $V$ and constraint family $\mathcal F$, is a set of pairs (weight, constraint) over $V$ and $\mathcal{F}$, where the weight is a positive rational number, denoted $P=\{\wcla{w_1}{f_1(x^1_{i_1},\dots,x^1_{i_{k_1}})},\dots,\wcla{w_m}{f_m(x^m_{i_1},\dots,x^m_{i_{k_m}})}\}$. The \emph{weight} of a problem is $\weight(P)=w_1+\cdots+w_m$.

An \emph{assignment} is a function $I:\{x_1,\dots,x_n\}\to\{0,1\}$. We say that an assignment $I$ \emph{satisfies} a constraint $f(x_{i_1},\dots,x_{i_k})$, if $I(f(\vec{x}))=_{\mbox{def}}f(I(x_{i_1}),\dots,I(x_{i_k}))=1$. 
The value of an assignment $I$ for a constraint problem $P=\{\wcla{w_i}{f_i(\vec{x})}\}_{i=1,\dots,m}$, is the sum of the weights of the constraints that this assignment satisfies, i.e. $I(P) = \sum_{i=1}^m w_i\,I(f_i(\vec{x}))$.
We also define the sum of weights of unsatisfied clauses as $\overline I(P) = \sum_{i=1}^m w_i\,(1-I(f_i(\vec{x})))$.
An assignment $I$ is said to be \emph{optimal} for a constraint problem $P$, if it maximizes $I(P)$. We define this optimum as $\opt(P)=\max_{I} I(P)$. We also define $\cost(P)=\min_I \overline I(P)$, i.e. the minimum sum of weights of falsified constraints, that fulfills $\opt(P)+\cost(P)=\weight(P)$.

\ignorar{
For convenience, we generalize constraint problems to multisets and consider $\{\wcla{w_1}{C},$ $\wcla{w_2}{C}\}$ as equivalent to $\{\wcla{w_1+w_2}{C}\}$. When the weight of a constraint is not explicitly specified, it is assumed to be one.

When using inference rules, we assume that they can be applied with any weight and that premises are decomposed conveniently. For instance, if we have the inference rule $C=0, C=1\vdash \eclause$, and we apply it to the problem $P=\{\wcla{w_1}{A=0},\,\wcla{w_2}{A=1}\}$, assuming $w_1\geq w_2$, we proceed as follows. First, we decompose the first constraint to get $P'=\{\wcla{w_2}{A=0},\,\wcla{w_1-w_2}{A=0},\,\wcla{w_2}{A=1}\}$. Then, we apply the rule with weight $w_2$ to get $P'=\{\wcla{w_2}{\eclause},\,\wcla{w_1-w_2}{A=0}\}$. When applying a rule, we assume that at least one of the premises is not decomposed.
}

Many optimization problems may be formalized as the problem of finding an optimal assignment for a constraint problem. In the following, we define some of them:
\begin{enumerate}
    \item MaxkSAT is the constraint family defined by the constraint functions of the form $f(x_1,\dots,x_r)= l_1\vee\cdots\vee l_r$, where every $l_i$ may be either $x_i$ or $\neg x_i$ and $r\leq k$.
    \item MaxkXOR is the constraint family defined by the constraint functions $f(x_1,\dots,x_r)=x_1\oplus\cdots\oplus x_r$ and $f(x_1,\dots,x_r)=x_1\oplus\cdots\oplus x_r\oplus 1$, where $r\leq k$.
\end{enumerate}

In this paper, we are interested in the problem Max2XOR, which has got very little attention in the literature, probably because it is quite similar to MaxCUT. 

Next, we define gadgets as a transformation from constraints into sets of weighted constraints (or problems). These transformations can be extended to define transformations of problems into problems.

\begin{definition}
Let $\mathcal{F}_1$ and $\mathcal{F}_2$ be two constraint families.
A \emph{$(\alpha,\beta)$-gadget} from $\mathcal{F}_1$ to $\mathcal{F}_2$ is a function that, for any constraint $f(\vec{x})$ over $\mathcal{F}_1$ returns a weighted constraint problem $P=\{\wcla{w_i}{g_i(\vec{x},\vec{b})}\}_{i=1,\dots,m}$ over $\mathcal{F}_2$ and variables $\{\vec{x}\}\cup\{\vec{b}\}$, where $\vec{b}$ are auxiliary variables distinct from $\vec{x}$, such that $\beta=\sum_{i=1}^m w_i$ and, for any assignment $I:\{\vec{x}\}\to\{0,1\}$:
\begin{enumerate}
    \item If $I(f(\vec{x}))=1$, for any extension of $I$ to $I':\{\vec{x}\}\cup\{\vec{b}\}\to\{0,1\}$, $I'(P)\leq \alpha$ and there exist one of such extension with $I'(P) = \alpha$.
    \item If $I(f(\vec{x}))=0$, for any extension of $I$ to $I':\{\vec{x}\}\cup\{\vec{b}\}\to\{0,1\}$, $I'(P)\leq \alpha-1$ and if, additionally,  there exists one of such extension with $I'(P) = \alpha-1$ we say the gadget is strict.
\end{enumerate}
\end{definition}

\begin{lemma}\label{lem-concatenation}
The composition of a $(\alpha_1,\beta_1)$-gadget from $\mathcal{F}_1$ to $\mathcal{F}_2$ and a $(\alpha_2,\beta_2)$-gadget from $\mathcal{F}_2$ to $\mathcal{F}_3$ results into a $(\beta_1\,(\alpha_2-1)+\alpha_1,\ \beta_1\beta_2)$-gadget from $\mathcal{F}_1$ to~$\mathcal{F}_3$.
\end{lemma}

\begin{proof}
The first gadget multiplies the total weight of constraints by $\beta_1$, and the second by $\beta_2$. Therefore, the composition multiplies it by $\beta=\beta_1\beta_2$.

For any assignment, if the original constraint is falsified, the optimal extension after the first gadget satisfies constraints with a weight of $\alpha_1-1$, and falsifies the rest $\beta_1-(\alpha_1-1)$. The second gadget satisfies constraints for a weight of $\alpha_2-1$ of the falsified plus $\alpha_2$ of the satisfied. Therefore, the composition satisfies constraints with a total weight $(\alpha_2-1)(\beta_1-(\alpha_1-1))+\alpha_2(\alpha_1-1)= \beta_1(\alpha_2-1)+\alpha_1-1$.

If the original constraint is satisfied, the optimal extension after the first gadget satisfies $\alpha_1$ and falsifies the rest $\beta_1-\alpha_1$. After the second gadget the weight of satisfied constraints is $(\alpha_2-1)(\beta_1-\alpha_1)+\alpha_2\alpha_1=\beta_1(\alpha_2-1)+\alpha_1$.

The difference between both situations is one, hence the composition is a gadget, and $\alpha=\beta_1(\alpha_2-1)+\alpha_1$.
\end{proof}


\section{Quantum Annealers as Max2XOR Solvers}

As we mention in the introduction, a quantum annealer can be seen as a Max2XOR solver. Roughly speaking, it is able to reach the state that minimizes the energy of the Ising Hamiltonian $\mathcal{H}=\sum_i h_i\pauli{i} +\sum_{(i,j)\in E} J_{ij}\pauli{i}\pauli{j}$, where $\pauli{i}$ is the z-Pauli operator acting on qubit $i$. For every linear term $h_i\pauli{i}$ of this Hamiltonian, we can get a 2XOR constraint $\wcla{2h_i}{x_i = 0}$, when $h_i>0$, or $\wcla{-2h_i}{x_i=1}$, when $h_i<0$. Similarly, for each quadratic term $J_{ij}\pauli{i}\pauli{j}$, we get a constraint $\wcla{2J_{ij}}{x_i\oplus x_j = 1}$, when $J_{ij}>0$, or $\wcla{-2J_{ij}}{x_i\oplus x_j=0}$, when $J_{ij}<0$. Conversely, any Max2XOR problem may be mapped into an Ising Hamiltonian.

In the following, we review with further detail the related work on constructing Hamiltonians for SAT and MaxSAT from the perspective of gadgets. 

When in \cite{quantumMax2SAT} they say that every 2SAT clause $x_i\vee x_j$ contributes to the Hamiltonian as $\frac{\mathbb{I}-\pauli{i}}{2}\frac{\mathbb{I}-\pauli{j}}{2}=\frac{1}{4}\mathbb{I}-\frac{1}{4}\pauli{i}-\frac{1}{4}\pauli{j}+\frac{1}{4}\pauli{i}\pauli{j}$, they are implicitly defining a $(1,3/2)$-gadget from 2SAT to 2XOR:
\begin{equation}
x_i\vee x_j \to 
\left\{\begin{array}{l}
\wcla{1/2}{x_i=1},\\
\wcla{1/2}{x_j=1},\\
\wcla{1/2}{x_i\oplus x_j=1}
\end{array}\right.
\label{eq-2SAT-2XOR}
\end{equation}

In the following, we only describe the translation of clauses $x_1\vee\cdots\vee x_k$ where all variables are positive. We can easily generalize the transformation when any of the variables $x$ is negated, simply by recalling that $\neg x\oplus y = k$ is equivalent to $x\oplus y = 1-k$, and $\neg x\oplus \neg y = k$ is equivalent to $x\oplus y=k$.

Notice also that we normalize the weights to ensure that the gap between the sum of weights when the original constraint is falsified ($\alpha-1$) and the sum when it is satisfied ($\alpha$) is just one. Later, when from the Max2XOR problem we construct the Hamiltonian, we can multiply the weights for the maximal value that still allows biases and couplings to be inside their corresponding ranges $[-1,1]$, while we try to maximize this gap. For example, in the case of the previous gadget $\{\wcla{1/2}{x_1=1},\ 
\wcla{1/2}{x_2=1},\ 
\wcla{1/2}{x_1\oplus x_2=1}\}$, when we transform it into a Hamiltonian $-\frac{1}{4}\pauli{1}-\frac{1}{4}\pauli{2}+\frac{1}{4}\pauli{1}\pauli{2}$ we could multiply all weights by $4$ and still getting all coefficients in the range $[-1,1]$. This allows us to enlarge the energy gap between the lowest-energy state and the next classical state to $4$. We call this quantity the \emph{energy gap} of the gadget, noted $\Delta E$. Remark that, as defined below, this energy gap only corresponds to the minimal difference of energy levels when we consider the translation of \emph{just one} clause. When more clauses are involved, if we want to keep biases and coupling inside their corresponding ranges, the situation is more complicated, obtaining smaller energy differences when normalizing. However, what we define as \emph{energy gap} is a good indicator of how good the gadget is, as it has been experimentally observed in \cite{cerquides}.



Formally,
\begin{definition}
Given a gadget $P$ to Max2XOR or, in general, a Max2XOR problem $P$, we define \emph{energy gap} as
\[
\Delta E = \min\left\{
\min_{i}\frac{1}{|h_i|}, 
\min_{i,j}\frac{1}{|J_{ij}|}
\right\}
\]
where, remember that
\[
\begin{array}{l}
h_i = \frac{1}{2}(\sum_{\wcla{w}{x_i=0}\in P}w-\sum_{\wcla{w}{x_i=1}\in P}w)\\
J_{ij}=\frac{1}{2}(\sum_{\wcla{w}{x_i\oplus x_j=1}\in P}w-\sum_{\wcla{w}{x_i\oplus x_j=0}\in P}w)
\end{array}
\]
Assuming that the problem $P$ is simplified, i.e. contains a unique constraint of the form $x_i=0$ or $x_i=1$, and a unique constraint of the form $x_i\oplus x_j =0$ or $x_i\oplus x_j = 1$, this simplifies to
\[
\Delta E = \min_{\wcla{w}{C}\in P}\ \frac{2}{w}
\]
\end{definition}

As we mentioned, N\"u\ss lein et al.~\cite{Nusslein} propose two ways to reduce Max3SAT to QUBO and compare them with previous methods. Here, we will only discuss the reduction that they call {\sc N\"usslein}$^{n+m}$, for which they report the best results (the other ones make use of more additional variables). In this reduction, each clause $x_1\vee x_2\vee x_3$ contributes to the QUBO problem as:
\[
2x_1x_2-2x_1b-2x_2b+x_3b-x_3+b
\]
Formalized as a contribution to an Ising Hamiltonian, this is:
\[
H=\frac{z_1z_2}{2}-\frac{z_1b}{2}-\frac{z_2b}{2}+\frac{z_3b}{4}-\frac{z_3}{4}-\frac{b}{4}
\]
We can formalize it as the $(5/2,9/2)$-gadget with $\Delta E = 2$ from 3SAT to Max2XOR represented in Figure~\ref{fig-nusslein}. Notice that the weights in the 2XOR constraints correspond to the coefficients of the Hamiltonian multiplied by $2$. 
In the rest of the paper, we will use this graphical representation where solid blue lines between $x$ and $y$ represent $x\oplus y=1$ and red dashed lines represent $x\oplus y=0$. Constraints $x=0$ and $x=1$ are graphically represented using lines to a special node $\hat{1}$.

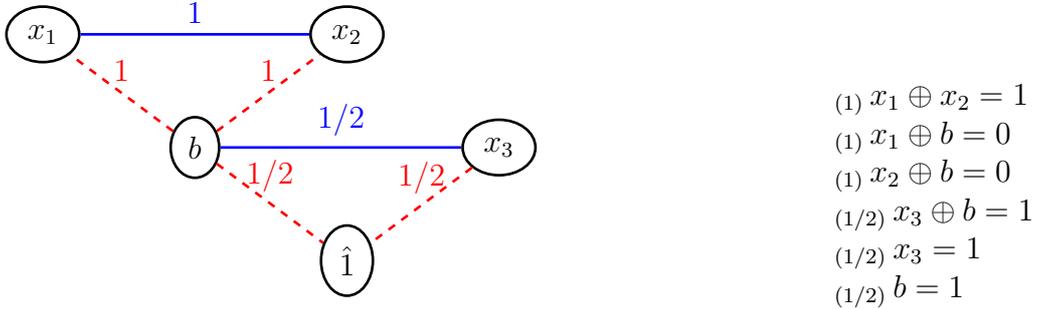
\begin{figure}[t]
\centering
\begin{tikzpicture}[xscale=2,yscale=1.5]
    \node[nodo] at (2,0) (1){$\hat{1}$};
    \node[nodo] at (0,2) (x1){$x_1$};
    \node[nodo] at (2,2) (x2){$x_2$};
    \node[nodo] at (3,1) (x3){$x_3$};
    \node[nodo] at (1,1) (b){$b$};
    \draw[arrowblue] (x1) --(x2) node[midway,above] {$1$};
    \draw[arrowblue] (b) -- (x3) node[midway,above] {$1/2$};
    \draw[arrowred] (b) -- (x1) node[midway,above] {$1$};
    \draw[arrowred] (b) -- (x2) node[midway,above] {$1$};
    \draw[arrowred] (b) -- (1) node[midway,above] {$1/2$};
    \draw[arrowred] (x3) -- (1) node[midway,above] {$1/2$};
    \end{tikzpicture}
\hfill$
\begin{array}[b]{l}
\wcla{1}{x_1\oplus x_2=1}\\
\wcla{1}{x_1\oplus b=0}\\
\wcla{1}{x_2\oplus b=0}\\
\wcla{1/2}{x_3\oplus b=1}\\
\wcla{1/2}{x_3=1}\\
\wcla{1/2}{b=1}
\end{array}
$
\caption{Graphic representation of the $(5/2,9/2)$-gadget with $\Delta E=2$ from 3SAT to Max2XOR proposed by N\"u\ss lein et al.~\cite{Nusslein}. Blue edges represent equal-one constraints, and red dashed edges are equal-zero constraints.}\label{fig-nusslein}
\end{figure}

Chancellor et al.~\cite{Chancellor} propose a method to translate kSAT to QUBO. It is based on the equivalence of the kSAT constraints $x_1\vee\cdots\vee x_k$ and the polynomial $1+\sum_{S\subseteq \{1,\dots,k\}} (-1)^{|S|+1}\prod_{i\in S} x_i$ and the use of a new qubit to represent terms of size bigger than two. There are still different ways to select the coupling factors. For a convenient election, we can get the gadget represented in Figure~\ref{fig-chancellor}.

\begin{figure}[t]
\centering
\begin{tikzpicture}[xscale=2,yscale=2]
    \node[nodo] at (4,3) (1){$\hat{1}$};
    \node[nodo] at (0,2) (x1){$x_1$};
    \node[nodo] at (2,2) (x2){$x_2$};
    \node[nodo] at (1,3.732) (x3){$x_3$};
    \node[nodo] at (1,2.577) (b){$b$};
    \draw[arrowblue] (x1) --(x2) node[midway,above] {$1/2$};
    \draw[arrowblue] (x1) --(x3) node[midway,left] {$1/2$};
    \draw[arrowblue] (x2) --(x3) node[midway,right] {$1/2$};
    \draw[arrowblue] (b) -- (x3) node[midway,left] {$1/2$};
    \draw[arrowblue] (b) -- (x1) node[midway,above] {$1/2$};
    \draw[arrowblue] (b) -- (x2) node[midway,above] {$1/2$};
    \draw[arrowred] (b) -- (1) node[midway,left] {$1/2$};
    \draw[arrowred] (x1) -- (1) node[midway,below] {$1/2$};
    \draw[arrowred] (x2) -- (1) node[midway,below] {$1/2$};
    \draw[arrowred] (x3) -- (1) node[midway,above] {$1/2$};
    \end{tikzpicture}
\hfill$
\begin{array}[b]{l}
\wcla{1/2}{x_1\oplus x_2=1}\\
\wcla{1/2}{x_1\oplus x_3=1}\\
\wcla{1/2}{x_2\oplus x_3=1}\\
\wcla{1/2}{x_1\oplus b=1}\\
\wcla{1/2}{x_2\oplus b=1}\\
\wcla{1/2}{x_3\oplus b=1}\\
\wcla{1/2}{x_1=1}\\
\wcla{1/2}{x_2=1}\\
\wcla{1/2}{x_3=1}\\
\wcla{1/2}{b=1}
\end{array}
$
\caption{Graphic representation of the $(3,5)$-gadget with $\Delta E=4$ from 3SAT to Max2XOR obtained with the method proposed by Chancellor et al.~\cite{Chancellor}.}\label{fig-chancellor}
\end{figure}

Bian et al.~\cite{quantumSebastianiConference,quantumSebastianiJournal} propose a method that can be used to reduce kSAT (or even more general formulas) to Max2XOR. It is based on the translation of each constraint of the form $x_1\vee x_2\leftrightarrow x_3$ as a contribution of
\[
\frac{z_1z_2}{2}-z_1z_3-z_2z_3+\frac{z_1}{2}+\frac{z_2}{2}-z_3+\frac{5}{2}
\]
to the Hamiltonian of a Ising model. Implicitly, it defines the non-strict $(3,9/2)$-gadget with $\Delta E=2$:
\begin{equation}
(x_1\vee x_2 \leftrightarrow b)\to
\left\{\begin{array}{l}
\wcla{1/2}{x_1=0}\\
\wcla{1/2}{x_2=0}\\
\wcla{1}{b=1}\\
\wcla{1/2}{x_1\oplus x_2=1}\\
\wcla{1}{x_1\oplus b=0}\\
\wcla{1}{x_2\oplus b=0}\\
\end{array}\right.
\label{eq-sebastiani}
\end{equation}
Notice that the values of the polynomial coefficients correspond to the weights of the constraints in the gadget. The independent term $5/2$ is only introduced in order to ensure that the ground-state energy is zero. Recall that this is just a convention, and the values of this term are not restricted to any range.

The application of the Tseitin encoding also can be seen as a gadget. For instance, the $(2,2)$-gadget to reduce 3SAT
\begin{equation}
x_1\vee x_2\vee x_3 \to 
\left\{\begin{array}{l}
x_1\vee x_2\leftrightarrow b\\
b\vee x_3
\end{array}\right.
\label{eq-tseiting}
\end{equation}
The composition of this gadget~(\ref{eq-tseiting}) with both (\ref{eq-2SAT-2XOR}) and (\ref{eq-sebastiani}), results in the $(4,6)$-gadget with $\Delta E =1$ from 3SAT to Max2XOR used in \cite{quantumSebastianiJournal} to translate 3SAT clauses into Hamiltonian components, and represented here in Figure~\ref{fig-sebastiani}.

\begin{figure}[t]
\centering
\begin{tikzpicture}[xscale=2,yscale=1.5]
    \node[nodo] at (2,0) (1){$\hat{1}$};
    \node[nodo] at (0,2) (x1){$x_1$};
    \node[nodo] at (2,2) (x2){$x_2$};
    \node[nodo] at (3,1) (x3){$x_3$};
    \node[nodo] at (1,1) (b){$b$};
    \draw[arrowblue] (x1) to [in=180, out = 270] (1) node at (0.6,0.6) {$1/2$};
    \draw[arrowblue] (x2) -- (1) node[midway,right] {$1/2$};
    \draw[arrowblue] (x1) --(x2) node[midway,above] {$1/2$};
    \draw[arrowblue] (b) -- (x3) node[midway,above] {$1/2$};
    \draw[arrowred] (b) -- (x1) node[midway,above] {$1$};
    \draw[arrowred] (b) -- (x2) node[midway,above] {$1$};
    \draw[arrowred] (b) -- (1) node[midway,above] {$3/2$};
    \draw[arrowred] (x3) -- (1) node[midway,above] {$1/2$};
    \end{tikzpicture}
\hfill$
\begin{array}[b]{l}
\wcla{1/2}{x_1=0}\\
\wcla{1/2}{x_2=0}\\
\wcla{1/2}{x_3=1}\\
\wcla{3/2}{b=1}\\
\wcla{1/2}{x_1\oplus x_2=1}\\
\wcla{1}{x_1\oplus b=0}\\
\wcla{1}{x_2\oplus b=0}\\
\wcla{1/2}{x_3\oplus b=1}
\end{array}
$
\caption{Graphic representation of the $(4,6)$-gadget with $\Delta E=4/3$ obtained by Tseitin encoding of a 3SAT clause, as proposed by Bian et al.~\cite{quantumSebastianiConference,quantumSebastianiJournal}.}\label{fig-sebastiani}
\end{figure}
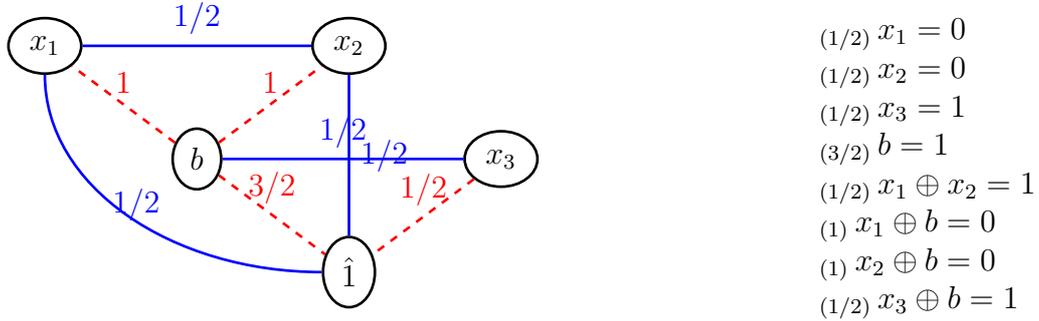

\section{New Gadgets From SAT to Max2XOR}\label{sec:MaxSAT->Max2XOR}

As we have seen in the previous section, the construction of Hamiltonians to solve a problem $P$ with a quantum annealer is basically the search for a gadget from the constraints in $P$ to Max2XOR. In general, there are three aspects that we would like to optimize in this gadget. First, we want to reduce the number of auxiliary variables (or ancillary variables in quantum terminology) because they imply using more qubits. Second, we want to relax the structure of the gadget, and in general avoid using gadgets with a dense structure (such as cliques), because there are limitations in the architecture of the quantum annealer, i.e., the graph of allowed couplings. Third, we want to maximize the energy gap, because this reduces the probability of errors.

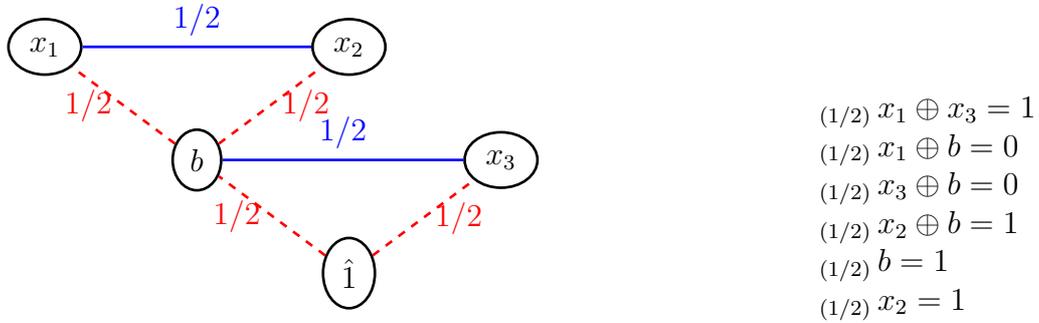
\begin{figure}[t]
\centering
\begin{tikzpicture}[xscale=2,yscale=1.5]
    \node[nodo] at (2,0) (1){$\hat{1}$};
    \node[nodo] at (0,2) (x1){$x_1$};
    \node[nodo] at (3,1) (x3){$x_3$};
    \node[nodo] at (2,2) (x2){$x_2$};
    \node[nodo] at (1,1) (b){$b$};
    \draw[arrowblue] (x3) -- (b) node[midway,above] {$1/2$};
    \draw[arrowblue] (x1) -- (x2) node[midway,above] {$1/2$};
    \draw[arrowred] (1) -- (x3) node[midway,right] {$1/2$};
    \draw[arrowred] (1) -- (b) node[midway,left] {$1/2$};
    \draw[arrowred] (b) -- (x1) node[midway,left] {$1/2$};
    \draw[arrowred] (b) -- (x2) node[midway,right] {$1/2$};
    \end{tikzpicture}
    \hfill
    $
    \begin{array}[b]{l}
    \wcla{1/2}{x_1 \oplus x_3 = 1}\\
    \wcla{1/2}{x_1 \oplus b = 0}\\
    \wcla{1/2}{x_3 \oplus b = 0}\\
    \wcla{1/2}{x_2 \oplus b = 1}\\
    \wcla{1/2}{b  = 1}\\
    \wcla{1/2}{x_2 = 1}\\
    \end{array}
    $
    \caption{Graphical representation of the $(2,3)$-gadget  with $\Delta E=4$  reducing Max3SAT to Max2XOR and based on Trevisan's~\cite{gadgets} gadget from Max3SAT to Max2SAT.}
    \label{fig:Max3SAT->Max2XOR}
\end{figure}

\subsection{A Gadget from 3SAT to Max2XOR}

An easy way to reduce SAT to Max2XOR is to reduce SAT to 3SAT, using the gadget~(\ref{eq-MaxSAT->Max3SAT}), this to Max2SAT, using the gadget~(\ref{eq-Max3SAT->Max2SAT})~\cite{gadgets,AnsoteguiL21},
and this to Max2XOR, using the gadget~(\ref{eq-2SAT-2XOR}).

The concatenation of the gadget~(\ref{eq-Max3SAT->Max2SAT}) and the gadget~(\ref{eq-2SAT-2XOR}) results into the $(2,3)$-gadget with $\Delta E=4$ described in Figure~\ref{fig:Max3SAT->Max2XOR}. This gadget is similar to the gadgets described in Figures~\ref{fig-nusslein}, \ref{fig-chancellor} and \ref{fig-sebastiani}.
However, this gadget is optimal in terms of minimizing $\alpha$, $\beta$, $\Delta E$ and the number of auxiliary variables. 

However, when we concatenate gadget~(\ref{eq-MaxSAT->Max3SAT}) with the gadget in Figure~\ref{fig:Max3SAT->Max2XOR} to obtain a reduction from kSAT to Max2XOR, for $k\geq 4$, we get a $(2(k-2),\,3(k-2))$-gadget with $2k-5$ auxiliary variables and $\Delta E=4$. This one is not optimal, in the sense that, as we will see below, there exist other gadgets that introduce fewer variables with the same $\Delta E$, and smaller $\alpha$ and $\beta$ parameters. 

In our experience, in general, if we concatenate two gadgets, even if both are optimal (in some parameter), the resulting gadget can be suboptimal (in that parameter). Therefore, it is usually better to compute \emph{direct} gadgets. In the following, we describe direct and better reductions (for some of the parameters) from SAT to Max2XOR.

\subsection{A Regular-Like Gadget}

Here, we present a new gadget inspired by the Refined-Regular gadget from SAT to Max2SAT introduced in~\cite{AnsoteguiL21}.
The reduction is based on a function $T^0$ that takes as parameters a SAT clause and a variable. The use of this variable is for technical reasons, to make simpler the recursive definition and proof of Lemma~\ref{lem:T0}. Later (see Theorem~\ref{thm-MaxSAT->PC}), it will be replaced by the constant one.

\begin{definition}
Given a SAT clause $x_1\vee\cdots\vee x_k$ and an auxiliary variable $b$, define the Max2XOR problem $T^0(x_1\vee\cdots\vee x_k,b)$ recursively as follows:
\[
T^0(x_1\vee x_2,b) = \left\{
\begin{array}{l}
\wcla{1/2}{x_1\oplus x_2=1}\\
\wcla{1/2}{x_1\oplus b=0}\\
\wcla{1/2}{b\oplus x_2=0}
\end{array}\right.
\]
for binary clauses, and 

\[
T^0(x_1\vee\cdots\vee x_k,b) = 
T^0(x_1\vee\cdots\vee x_{k-1},b')\cup
\left\{\begin{array}{l}
\wcla{1/2}{b'\oplus x_k=1}\\
\wcla{1/2}{b'\oplus b=0}\\
\wcla{1/2}{b\oplus x_k=0}
\end{array}\right.
\]

\noindent for $k\geq 3$.
\end{definition}

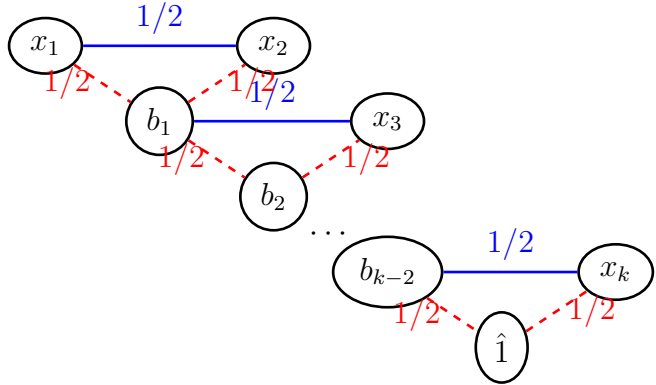
\begin{figure}[t]
\centering
\begin{tikzpicture}[xscale=1.5]
\node[nodo] at (0,4) (x1){$x_1$};
\node[nodo] at (1,3) (b1){$b_1$};
\node[nodo] at (2,4) (x2){$x_2$};
\node[nodo] at (2,2) (b2){$b_2$};
\node[nodo] at (3,3) (x3){$x_3$};
\node[nodo] at (3,1) (bn2){$b_{k-2}$};
\node[nodo] at (4,0) (bn1){$\hat{1}$};
\node[nodo] at (5,1) (xn){$x_k$};
\node at (2.5,1.5) {$\cdots$};
\draw[arrowblue] (x1) -- (x2) node[midway,above] {$1/2$};
\draw[arrowblue] (b1) -- (x3) node[midway,above] {$1/2$};
\draw[arrowblue] (bn2) -- (xn) node[midway,above] {$1/2$};
\draw[arrowred] (x1) -- (b1) node[midway,left] {$1/2$};
\draw[arrowred] (b1) -- (b2) node[midway,left] {$1/2$};
\draw[arrowred] (bn2) -- (bn1) node[midway,left] {$1/2$};
\draw[arrowred] (b1) -- (x2) node[midway,right] {$1/2$};
\draw[arrowred] (b2) -- (x3) node[midway,right] {$1/2$};
\draw[arrowred] (bn1) -- (xn) node[midway,right] {$1/2$};
\end{tikzpicture}
\caption{Graphical representation of $T^0(x_1\vee\cdots\vee x_k, \hat{1})$, defining a $(k-1,3/2(k-1)$-gadget with $\Delta E=4$ form kSAT to Max2XOR that introduces $k-2$ variables. For $k=3$, this gadget corresponds to the one in Figure~\ref{fig:Max3SAT->Max2XOR}.}\label{fig-MaxSAT->PC}
\end{figure}

\begin{lemma}\label{lem:T0}
Consider the Max2XOR problem $T^0(x_1\vee\cdots\vee x_k, b_{k-1})$. We have:
\begin{enumerate}
\item For any assignment $I:\{x_1,\dots,x_k\}\to\{0,1\}$,
the extension of the assignment as $I'(b_{i})=I(x_{i+1})$, for $i=1,\dots,k-1$, maximizes the number of satisfied 2XOR constraints, that is $2\,(k-1)$.
\item For any assignment $I:\{x_1,\dots,x_k,b_{k-1}\}\to\{0,1\}$ satisfying $I(b_{k-1})=1$, the extension of the assignment  as
\[
I'(b_{i})=\left\{
\begin{array}{ll}
1 & \mbox{if $I(x_{i+2})=\cdots=I(x_k)=0$}\\[2mm]
I(x_{i+1}) &\mbox{otherwise}
\end{array}
\right.
\]
for $i=1,\dots,k-2$, maximizes the number of satisfied 2XOR constraints, that is $2\,(k-2)$, if $I(x_1)=\cdots=I(x_k)=0$, or $2\,(k-1)$ otherwise.
\end{enumerate}
\end{lemma}
\begin{proof}
See Lemma~\ref{lem:parallel}, where a more general result is proved.
\end{proof}

\begin{theorem}\label{thm-MaxSAT->PC}
The translation of every clause $x_1\vee\cdots\vee x_k$ into $T^0(x_1\vee\cdots\vee x_k, \hat{1})$ defines a $(k-1,3/2(k-1))$-gadget from kSAT to Max2XOR that introduces $k-2$ auxiliary variables and has $\Delta E=4$.
\end{theorem}
\begin{proof}
The soundness of the reduction is based on the second statement of Lemma~\ref{lem:T0}. Assume that we force the variable $\hat{1}$ to be interpreted as one, i.e. $I(\hat{1})=1$, or that, alternatively, we simply replace it by the constant $1$. Any assignment satisfying the clause can extend to satisfy at least $2(k-1)$ 2XOR constraints from $T^0(x_1\vee\cdots\vee x_k, \hat{1})$, with a total weight equal to $k-1$. Any assignment  falsifying it can only be extended to satisfy at most $2(k-2)$ 2XOR constraints, with a total weight $k-2$. Hence, $\alpha = k-1$. The number of 2XOR constraints is $3(k-1)$, all of them with weight $1/2$. Therefore, $\beta=3/2(k-1)$.
\end{proof}

\subsection{A Tree-Like Gadget}

In this section, we describe the Tree-like gadget, which is a generalization of the Regular-like gadget. Before describing the new gadget, we introduce an example for the clause $x_1\vee\cdots\vee x_7$, in Figure~\ref{fig-example}.

\begin{figure}[t]
\centering
\begin{tikzpicture}[xscale=1.5]
\node[nodo] at (0,3) (x1){$x_1$};
\node[nodo] at (1,3) (x2){$x_2$};
\node[nodo] at (2,3) (x3){$x_3$};
\node[nodo] at (3,3) (x4){$x_4$};
\node[nodo] at (4,3) (x5){$x_5$};
\node[nodo] at (5,3) (x6){$x_6$};
\node[nodo] at (6.5,2) (x7){$x_7$};
\node[nodo] at (0.5,2) (b1){$b_1$};
\node[nodo] at (2.5,2) (b2){$b_2$};
\node[nodo] at (4.5,2) (b3){$b_3$};
\node[nodo] at (1.5,1) (b4){$b_4$};
\node[nodo] at (5.5,1) (b5){$b_5$};
\node[nodo] at (3.5,0) (1){$\hat{1}$};
\draw[arrowblue] (x1) -- (x2);
\draw[arrowblue] (x3) -- (x4);
\draw[arrowblue] (x5) -- (x6);
\draw[arrowblue] (b1) -- (b2);
\draw[arrowblue] (b3) -- (x7);
\draw[arrowblue] (b4) -- (b5);
\draw[arrowred] (x1) -- (b1);
\draw[arrowred] (x2) -- (b1);
\draw[arrowred] (x3) -- (b2);
\draw[arrowred] (x4) -- (b2);
\draw[arrowred] (x5) -- (b3);
\draw[arrowred] (x6) -- (b3);
\draw[arrowred] (b1) -- (b4);
\draw[arrowred] (b2) -- (b4);
\draw[arrowred] (b3) -- (b5);
\draw[arrowred] (x7) -- (b5);
\draw[arrowred] (b4) -- (1);
\draw[arrowred] (b5) -- (1);
\end{tikzpicture}
\caption{A graphical representation of a Tree-like gadget for 7SAT. All constraints have weight $1/2$.}\label{fig-example}
\end{figure}
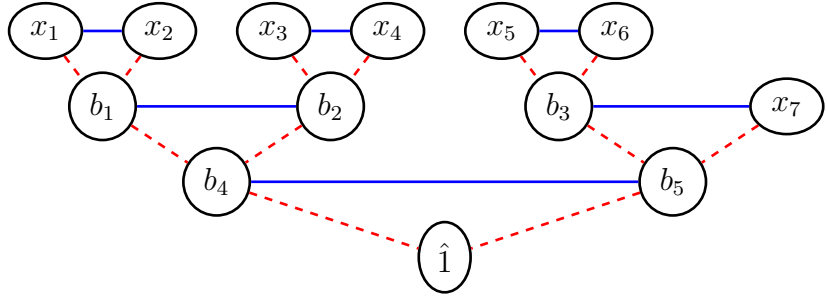

\ignorar{
By negating some of the $b$'s and $x$'s we can get an equivalent 2PC problem with only blue edges:

\begin{center}
\begin{tikzpicture}
\node[nodo] at (0,3) (x1){$\neg x_1$};
\node[nodo] at (1,3) (x2){$\neg x_2$};
\node[nodo] at (2,3) (x3){$\neg x_3$};
\node[nodo] at (3,3) (x4){$\neg x_4$};
\node[nodo] at (4,3) (x5){$\neg x_5$};
\node[nodo] at (5,3) (x6){$\neg x_6$};
\node[nodo] at (6.5,2) (x7){$x_7$};
\node[nodo] at (0.5,2) (b1){$b_1$};
\node[nodo] at (2.5,2) (b2){$b_2$};
\node[nodo] at (4.5,2) (b3){$b_3$};
\node[nodo] at (1.5,1) (b4){$\neg b_4$};
\node[nodo] at (5.5,1) (b5){$\neg b_5$};
\node[nodo] at (3.5,0) (1){$\hat{1}$};
    \draw[arrowblue] (x1) -- (x2);
    \draw[arrowblue] (x3) -- (x4);
    \draw[arrowblue] (x5) -- (x6);
    \draw[arrowblue] (b1) -- (b2);
    \draw[arrowblue] (b1) -- (b4);
    \draw[arrowblue] (b2) -- (b4);
    \draw[arrowblue] (x1) -- (b1);
    \draw[arrowblue] (x2) -- (b1);
    \draw[arrowblue] (x3) -- (b2);
    \draw[arrowblue] (x4) -- (b2);
    \draw[arrowblue] (x5) -- (b3);
    \draw[arrowblue] (x6) -- (b3);
    \draw[arrowblue] (x7) -- (b5);
    \draw[arrowblue] (b5) -- (1);
    \draw[arrowred] (b3) -- (x7);
    \draw[arrowred] (b4) -- (b5);
    \draw[arrowred] (b3) -- (b5);
    \draw[arrowred] (b4) -- (1);
\end{tikzpicture}
\end{center}
}

\begin{definition}
Given a clause $x_1\vee\cdots\vee x_k$ and an auxiliary variable $b$, define the Max2XOR problem $T^t(x_1\vee\cdots\vee x_k,b)$ recursively and non-deterministically as follows:
\[
T^t(x_1\vee\cdots\vee x_k,\ b) =
\tabcolsep 0mm
\left\{
\begin{tabular}{l@{\ }r}
$\{x_1\oplus x_2=1, x_1\oplus b=0, x_2\oplus b=0\}$
& if $k=2$\\[3mm]
\arraycolsep 0mm
$\begin{array}[t]{l}
T^t(x_1\vee\cdots\vee x_{k-1},\ b')\ \cup\\
\{b'\oplus x_k=1, b'\oplus b=0, x_k\oplus b=0\}

\end{array}$
&if $k\geq 3$\\[7mm]
\arraycolsep 0mm
$\begin{array}[t]{l}
T^t(x_2\vee\cdots\vee x_{k},\ b')\ \cup\\
\{b'\oplus x_1=1, b'\oplus b=0, x_1\oplus b=0\}
\end{array}$
&if $k\geq 3$\\[7mm]
\arraycolsep 0mm
$\begin{array}[t]{l}
\ T^t(x_1\vee\cdots\vee x_r,\ b')\ \cup\\
\ T^t(x_{r+1}\vee\cdots\vee x_k,\ b'')\ \cup\\
\{b'\oplus b''=1, b'\oplus b=0, b''\oplus b=0\}
\end{array}$
& 
\tabcolsep 0mm
\begin{tabular}[t]{l}
if $k\geq 4$, where \\ 
$2\leq r \leq k-2$
\end{tabular}
\end{tabular}
\right.
\]

\noindent where all XOR constraints have weight $1/2$.
\end{definition}

Notice that the definition of $T^t$ is not deterministic. When $k=3$ we have two possible translations (applying the second or third case), and when $k\geq 4$ we can apply the second, third, and fourth cases, with distinct values of $r$, getting as many possible translations as possible binary trees with $k$ leaves.

Notice also that the definition of $T^0$ corresponds to the definition of $T^t$ where we only apply the first and second cases. Therefore, $T^t$ is a generalization of $T^0$.

\begin{lemma}\label{lem:parallel}
Consider the set of 2XOR constraints of $T^t(x_1\vee\cdots\vee x_k,b)$ without weights. We have:
\begin{enumerate}
    \item There are $3(k-1)$ constraints.
    \item Any assignment satisfies at most $2(k-1)$ of them.
    \item Any assignment $I:\{x_1,\dots,x_k\}\to\{0,1\}$ can be extended to an optimal assignment satisfying $I(b)=I(x_1\vee\cdots\vee x_k)$ and $2(k-1)$ constraints.
    \item Any assignment $I$ satisfying $I(x_1\vee\cdots\vee x_k)=0$ and $I(b)=1$ can be extended to an optimal assignment satisfying $2(k-2)$ constraints.
\end{enumerate}
\end{lemma}

\begin{proof}
The first statement is trivial.

For the second, notice that for, each one of the $k-1$ \emph{triangles} of the form $\{a\oplus b=1, a\oplus c=0, b\oplus c=0\}$,  any assignment can satisfy at most two of the constraints of each triangle.

The third statement is proved by induction. The base case, for binary clauses, is trivial. For the induction case, if we apply the second or third options of the definition of $T$, the proof is similar to Lemma~\ref{lem:T0}. In the fourth case, assume by induction that there is an assignment extending $I:\{x_1,\dots,x_r\}\to\{0,1\}$ that verifies $I(b')=I(x_1\vee\cdots\vee x_r)$ and satisfies $2(r-2)$ constraints of $T^t(x_1\vee\cdots\vee x_r,b')$. Similarly, there is an assignment extending $I:\{x_{r+1},\dots,x_k\}\to\{0,1\}$ that verifies $I(b'')=I(x_{r+1}\vee\cdots\vee x_k)$ and satisfies $2((k-r)-2)$ constraints of $T^t(x_{r+1}\vee\cdots\vee x_k,b'')$. Both assignments do not share variables, therefore, they can be combined into a single assignment. This assignment can be extended with $I(b)=b'\vee b''$. This ensures that it will satisfy two of the constraints from $\{b'\oplus b''=1, b'\oplus b=0, b''\oplus b=0\}$. Therefore, it verifies $I(b)=I(x_1\vee\cdots\vee x_k)$ and satisfies $2(r-2)+2((k-r)-2)+2=2(k-1)$ constraints of $T^t(x_1\vee\cdots\vee x_k,b)$. Since this is the maximal number of constraints we can satisfy, this assignment is optimal.

The fourth statement is also proved by induction. The base case, for binary clauses, is trivial. For the induction case, consider only the application of the fourth case of the definition of $T$ (the other cases are quite similar). We have to consider 3 possibilities:

If we extend $I(b')=I(b'')=1$, by induction, we can extend $I$ to satisfy $2(r-2)$ constraints of $T^t(x_1\vee\cdots\vee x_r,b')$ and $2((k-r)-2)$ constraints of $T^t(x_{r+1}\vee\cdots\vee x_k,b'')$. It satisfies $2$ constraints from $\{b'\oplus b''=1, b'\oplus b=0, b''\oplus b=0\}$, hence a total of $2(k-3)$ constraints.

If we extend $I(b')=I(b'')=0$, by the third statement of this lemma, we can extend $I$ to satisfy $2(r-1)$ constraints of $T^t(x_1\vee\cdots\vee x_r,b')$ and $2((k-r)-1)$ constraints of $T^t(x_{r+1}\vee\cdots\vee x_k,b'')$. It does not satisfy any constraints from $\{b'\oplus b''=1, b'\oplus b=0, b''\oplus b=0\}$, hence a total of $2(k-2)$ constraints.

Finally, If we extend $I(b')=1$ and $I(b'')=0$ (or vice versa), by induction, we can extend $I$ to satisfy $2(r-2)$ constraints of $T^t(x_1\vee\cdots\vee x_r,b')$ and by the third statement $2((k-r)-2)$ constraints of $T^t(x_{r+1}\vee\cdots\vee x_k,b'')$. It also satisfy $2$ constraints from $\{b'\oplus b''=1, b'\oplus b=0, b''\oplus b=0\}$, hence a total of $2(k-2)$ constraints.
\end{proof}

\begin{theorem}\label{thm:parallel}
The translation of every clause $x_1\vee\cdots\vee x_k$
into $T^t(x_1\vee\cdots\vee x_k,\hat{1})$ 
defines a $(k-1, 3/2(k-1))$-gadget from kSAT to Max2XOR that introduces $k-2$ auxiliary variables, with $\Delta E=4$.
\end{theorem}

\begin{proof}
The theorem is a direct consequence of the two last statements of Lemma~\ref{lem:parallel}.
\end{proof}

If we compare Theorems~\ref{thm-MaxSAT->PC} and~\ref{thm:parallel}, we observe that both establish the same values for $\alpha=k-1$ and $\beta=3/2(k-1)$ in the $(\alpha,\beta)$-gadget, and both introduce the same number $(k-2)$ of auxiliary variables. Therefore, a priory, the parallel translation is a generalization of the sequential translation, but it does not have a clear (theoretical) advantage. 

However, in the context of quantum annealers, we recall that not all possible couplings are allowed, only the ones in the graph defined by the architecture. To circumvent this issue the typical approach is to make \emph{copies} (using extra qubits) of the variables involved in the Max2XOR constraint whose coupling is not allowed. In other words, if we have a constraint $\wcla{w}{x_i\oplus x_j=1}$ in the Max2XOR problem, but $(i,j)\not \in E$, we have to find a path and replace the constraint by $\{\wcla{1}{x_i \oplus x_{k_1}=0}, \dots, \wcla{1}{x_{k_{r-1}}\oplus x_r=0}, \wcla{w}{x_r \oplus x_j = 1}\}$

In this sense, the more \emph{flexible} structure of the parallel translation can help to fit more easily (i.e. with fewer extra qubits) the Max2XOR constraints generated by the gadget into the architecture of the quantum annealer.

\ignorar{ 
However, if we think of industrial instances with high modularity, it makes sense to generate a tree (like the one in Example~\ref{ex:parallelTranslation}), where distant variables in the original SAT problem, are located distantly in the tree. For instance, using the algorithms for analyzing the modularity of industrial SAT instances (see \cite{SAT12,JAIR19}) we can determine how distant are two variables in the original SAT problem, or even better, we can obtain a dendrogram of these variables. Then, we can use this distance or the dendrogram to construct a better translation, i.e. a tree where distant variables are kept in distinct branches. This will have the effect of increasing the modularity of the problem, hence decreasing its difficulty. We think that by using these trees for industrial SAT instances, we will obtain better experimental results than with the sequential translation.
}

\begin{figure}[t]
\centering
\hfill
    \begin{tikzpicture}[xscale=1.5, yscale=2]
    \node at (2,2){$\cdots$};
    \node[nodo] at (2,0) (1){$\hat{1}$};
    \node[nodo] at (1,2) (x1){$x_1$};
    \node[nodo] at (3,2) (x5){$x_5$};
    \node[nodo] at (0,1) (b1){$b_1$};
    \node[nodo] at (4,1) (b2){$b_2$};
    \draw[arrowblue] (x1) to[in=155,out=25] (x5); \node[arrowblue] at (2,2.4) {$1/2$};
    \draw[arrowblue] (b1) -- (x1) node[midway,left] {$5/6$};
    \draw[arrowblue] (b1) -- (x5) node[midway,above] {$5/6$};
    \draw[arrowblue] (b2) -- (x1) node[midway,above] {$2/3$};
    \draw[arrowblue] (b2) -- (x5) node[midway,right] {$2/3$};
    \draw[arrowblue] (b1) -- (b2) node[midway,above] {$5/6$};
    \draw[arrowred] (x1) -- (1) node[midway,below] {$1/2$};
    \draw[arrowred] (x5) -- (1) node[midway,below] {$1/2$};
    \draw[arrowred] (b1) -- (1) node[midway,below] {$2/3$};
    \draw[arrowred] (b2) -- (1) node[midway,below] {$5/6$};
    \end{tikzpicture}
\hfill$
\begin{array}[b]{l}
\wcla{1/2}{x_i=1}\\
\wcla{2/3}{b_1=1}\\
\wcla{5/6}{b_2=1}\\
\wcla{1/2}{x_i\oplus x_j=1}\\
\wcla{5/6}{x_i\oplus b_1=1}\\
\wcla{2/3}{x_i\oplus b_2=1}\\
\wcla{5/6}{b_1\oplus b_2=1}\\[2mm]
\mbox{where $i,j=1,\dots,5$}\\
\mbox{and $i\neq j$}
\end{array}
$
\caption{Graphic representation of the clique-like $(10,52/3)$-gadget from 5SAT to Max2XOR, with $\Delta E=12/5$, and optimal number of $2$ auxiliary variables.}\label{fig-cliquek5}
\end{figure}

\subsection{A Clique-Like Gadget}

The previous gadgets have an energy gap $\Delta E= 4$, but they are not optimal with respect to the number of auxiliary variables. We can obtain gadgets that only require $\mathcal{O}(\log k)$ auxiliary variables, on expenses of decreasing the energy gap and increasing the density of constraints. Therefore, there is a trade-off between optimizing the energy gap and the number of auxiliary variables.

\begin{definition}
Given a clause $x_1\vee\cdots\vee x_k$, where $k$ is a power of $2$, define the Max2XOR problem $T^c(x_1\vee\cdots\vee x_k)$ as follows:
\[
T^c(x_1\vee\cdots\vee x_k)=
\left\{
\begin{array}{l@{\hspace{5mm}}l}
\wcla{1/2}{x_i = 1} & i=1,\dots, k\\
\wcla{2^{j-1}}{b_j = 1} & j=1,\dots, \log k-1\\
\wcla{1/2}{x_i\oplus x_j = 1} & 1\leq i<j\leq k\\
\wcla{2^{i+j-1}}{b_i\oplus b_j=1} & 1\leq i<j\leq \log k-1\\
\wcla{2^{j-1}}{x_i\oplus b_j=1} & i=1,\dots, k\\ &j=1,\dots,\log k-1\\
\end{array}\right.
\]
\end{definition}
\begin{theorem}
The transformation of every clause $x_1\vee\cdots\vee x_k$ into $T^c(x_1\vee\cdots\vee x_k)$ defines a $(\alpha,\beta)$-gadget from kSAT to Max2XOR that introduces $\log k-1$ auxiliary variables, with
\[
\begin{array}{l}
\displaystyle\alpha=\frac{1}{2}\,k\,(k-1)\\[2mm]
\displaystyle\beta=\frac{11\,k^2-15\,k+4}{12}\\[2mm]
\displaystyle\Delta E = 2^5\,k^{-2}
\end{array}
\]
\end{theorem}
\begin{proof}
The intuition for the construction of the gadget is: apart from the variables $x_i$ (with weight $1$), we add auxiliary variables $b_j$ with weight $2^j$ and a constant $\hat{0}$ with weight $1$. The sum of all weights is $2k-1$. The weight of every constraint $x\oplus y =1$ is equal to one-half of the weight of $x$ by the weight of $y$. No matter how $x$'s are assigned, if at least one is assigned to one, we can assign the $b$'s such that the weight of variables assigned to one is $k$ and the weight of the assigned to zero is $k-1$, or vice versa. This maximizes the weight of satisfied assignments that is $k(k-1)/2=\alpha$. Conversely, if all $x$'s are assigned to zero, the best we can do is assign all $b$'s to one, which results in weight $(k+1)(k-2)/2=\alpha-1$ for the satisfied constraints.

More formally, let $I$ be an assignment satisfying $x_1\vee\cdots\vee x_k$. Let $r=\sum_{i=1}^k I(x_i)\in[1,k]$. We extend $I$ such that $I(b_i)$ is the $i$th bit of the binary representation of the number $k-r$. Therefore,
\(
\sum_{i=1}^{\log k-1} I(b_i)\, 2^i 
\) is either $k-r$ or $k-r-1$. And the sum of weights of variables assigned to one,
\(
w_1\stackrel{def}{=}\sum_{i=1}^{\log k-1} I(b_i)\, 2^i + \sum_{i=1}^k x_i
\) is either $k$ or $k-1$. For the ones assigned to zero (including $\hat{0}$) we have
\(w_0 \stackrel{def}{=} \sum_{i=1}^{\log k-1} (1-I(b_i))\, 2^i + \sum_{i=1}^k (1-I(x_i)) +1=2k-1-w_1
\). Therefore, the sum of satisfied constraints is $w_1\, w_0/2 = k(k-1)/2$. 

If $I$ falsifies the clause, we can set $I(b_i)=1$ and get $w_1 = \sum_{i=1}^{\log k-1} 2^i = k-2$ and $w_0=k+1$. In this case, we get satisfied constraints for a total weight $(k+1)(k-2)/2 = k(k-1)/2 -1$.
\end{proof}

When $k$ is not a power of $2$, the gadget with a minimal number of auxiliary variables is not straightforward to generalize. In Figure~\ref{fig-cliquek5}, we represent the gadget for $k=5$ with $2$ auxiliary variables.

\section{Computing Gadgets Automatically}

As in \cite{AnsoteguiL21}, we have created a Mixed Integer Programming model to compute automatically gadgets from SAT to Max2XOR with the MIP solver Gurobi.

Given a $k$SAT clause and a set of new auxiliary variables $b_j$ the MIP program explores the subsets of weighted XOR clauses of length at most 2 that can be constructed with the $k$ input variables and the $b_j$ variables. The weights of the XOR clauses are represented with continuous variables in the interval $[0,1]$. 

The MIP model incorporates constraints to ensure that the selected subset of XOR clauses fulfills the definition of a  $(\alpha,\beta)$-gadget, and the objective function maximizes the energy gap $\Delta E$.  The MIP model is solved to optimality.

In the following table, we report the $\alpha$, $\beta$ and $\Delta E$ values found with the MIP program for some combinations of $k$ and number of auxiliary variables:

\begin{center}
\begin{tabular}{|c|c|c|c|}
\hline
&\multicolumn{3}{|c|}{Number of aux. vars.}\\
\hline
 & 1 & 2 & 3 \\
\hline
k=3 & 
$\begin{array}{l}\alpha=2\\ \beta=3 \\ \Delta E=4 \end{array}$ &&\\
\hline
k=4 &
$\begin{array}{l}\alpha=6\\ \beta=10 \\ \Delta E=2 \end{array}$ &
$\begin{array}{l}\alpha=3\\ \beta=9/2 \\ \Delta E=4 \end{array}$ &
\\
\hline
k=5 &
- &
$\begin{array}{l}\alpha=10\\ \beta=52/3 \\ \Delta E= 12/5\end{array}$ &
$\begin{array}{l}\alpha=4\\ \beta=6 \\ \Delta E=4 \end{array}$ \\
\hline
\end{tabular}
\end{center}

As we can see, the gadgets found by the MIP program presented in the diagonal of the table have the same values as the Tree-like gadget, while the gadgets under the diagonal have the same values as the Clique-like gadget, certifying that, for small values of $k$, they are optimal respect to the energy gap and the number of auxiliary variables, respectively.


\section{Conclusions}
Quantum annealers \emph{may} offer in the future an alternative to solve the SAT problem more efficiently. To take maximum advantage of this new computation architecture, we need to take into account several aspects, such as the number of qubits that needs our reformulation of the SAT problem, the energy gap, and how adaptable is our encoding to the graph describing the allowed couplings of the qubits in the quantum annealer (what may incur into the use of additional auxiliary variables).

Tree-like gadgets help us to maximize the energy gap using a linear number of auxiliary variables and are rather \emph{flexible} to adapt them efficiently to different graph architectures. On the other hand, Clique-like gadgets require a logarithmic number of auxiliary variables, (i.e. fewer qubits), but with a lower energy gap, and may use additional qubits to accommodate the dense structure of the gadget to the architecture of the quantum annealer.

\section*{Acknowledgements}
This work was supported by
MCIN/AEI/10.13039/501100011033 \textit{(Grant: PID2019-109137GB-C21)}.



\bibliographystyle{elsarticle-num-names} 
\bibliography{ijar}





\end{document}